\documentclass[letterpaper, 10pt, conference]{IEEEtran}  

\IEEEoverridecommandlockouts                              

\usepackage{amsthm} 
\usepackage{amsmath,amssymb,amsfonts,graphicx,xcolor,cite,algorithm}
\usepackage{cancel}
\usepackage[margin=1in]{geometry}
\usepackage[noend]{algpseudocode}

\newtheorem{remark}{Remark}
\theoremstyle{definition}
\newtheorem{definition}{Definition}
\newtheorem{theorem}{Theorem}
\newtheorem{prop}{Proposition}
\newtheorem{corollary}{Corollary}

\newcommand{\real}{\mathcal{R}}

\newcommand{\tr}[1]{\textbf{tr}\left(#1\right)}
\newcommand{\Exp}[1]{\mathbb{E}\left[#1\right]}
\newcommand{\Var}[1]{\Exp{#1{#1}^T}}

\newcommand{\eqn}[1]{(\ref{eqn:#1})}
\newcommand{\eqnlabel}[1]{\label{eqn:#1}}

\title{\LARGE \bf
Optimal State Estimation in the Presence of Non-Gaussian Uncertainty via Wasserstein Distance Minimization
}

\author{Himanshu Prabhat and Raktim Bhattacharya
\thanks{Himanshu Prabhat and Raktim Bhattacharya are with the Department of Aerospace Engineering, Texas A\&M University, TX, USA. Their email addresses are {\tt\small himanshu.pr007@tamu.edu} and {\tt\small raktim@tamu.edu} respectively.}%
}

\begin{document}

\maketitle
\thispagestyle{empty}
\pagestyle{empty}

\begin{abstract}
This paper presents a novel distribution-agnostic Wasserstein distance-based estimation framework. The goal is to determine an optimal map combining prior estimate with measurement likelihood such that posterior estimation error optimally reaches the Dirac delta distribution with minimal effort. The Wasserstein metric is used to quantify the effort of transporting from one distribution to another. We hypothesize that minimizing the Wasserstein distance between the posterior error and the Dirac delta distribution results in optimal information fusion and posterior state uncertainty. Framework validation is demonstrated by the successful recovery of the classical Kalman filter for linear systems with Gaussian uncertainties. Notably, the proposed Wasserstein filter does not rely on particle representation of uncertainty. Furthermore, the classical result for the Gaussian Sum Filter (GSF) is retrieved from the Wasserstein framework. This approach analytically exhibits the suboptimality of GSF and enables the use of nonlinear optimization techniques to enhance the accuracy of the Gaussian sum estimator.
\end{abstract}

\section{INTRODUCTION}
The necessity for nonlinear estimation algorithms arises from various challenges inherent to complex systems and their associated dynamics. Nonlinear dynamics often evolve in nonlinear manifolds, necessitating the development of estimation techniques that can capture these complexities. Sensor models are frequently nonlinear and exhibit time-varying, state-dependent noise, complicating the estimation process further. While linearized dynamics with frequent measurements have proven effective in practice, they tend to fail when measurements are sparse spatio-temporally. In such cases, nonlinearities in the dynamics dominate system trajectories, resulting in non-Gaussian state uncertainties that cannot be resolved using linear approximations. Particle-based methods have been proposed to address these issues; however, they are often plagued with concerns related to unbiasedness, lack of convergence guarantees, and computational inefficiencies.

As a result, nonlinear estimation remains an active area of research, with numerous novel approaches emerging in the literature. Recently, the use of optimal transport has been gaining popularity in the field of nonlinear estimation. Developments in this area show promise, offering potential solutions to the challenges posed by nonlinear dynamics and complex sensor models. By incorporating these techniques, researchers are advancing the state of the art in nonlinear estimation, enhancing its accuracy and reliability in a wide range of applications.

The first challenge in nonlinear filtering is the propagation of uncertainty to get the prior state uncertainty, which is governed by advection-diffusion equations. Wasserstein gradient flows provide a powerful means of understanding and solving many diffusion equations. Specifically, Fokker-Planck equations, which model the diffusion of probability measures, can be understood as gradient descent over entropy functionals in Wasserstein space. This simplifies the computation of the prior uncertainty in nonlinear filtering \cite{mokrov2021large, frogner2020approximate, shafieezadeh2018wasserstein}. 

The second challenge is determining the optimal posterior state uncertainty with available measurements. This step is difficult as the prior and the measurement uncertainties can potentially be in different manifolds, represented with a mix of Cartesian and circular variables. Information fusion in the optimal transport framework is primarily performed using particle representation of distributions. Danu et al. \cite{danu2009wasserstein} treat the problem of evaluating the cost of particle filter clouds association based on the Wasserstein distance of different orders, analyzing the implications of clouds cardinality (for weighted particles), and of various resampling methods (for unweighted particles). More recently, Reich et al. \cite{reich2013ensemble, reich2013nonparametric} leveraged the optimal transport property of Wasserstein distance to resample weighted point clouds from Bayesian updates and generate unweighted particles. The formulation works quite well in nonlinear manifolds, as shown by Das et al. \cite{das2019optimal, das2019optimal2}. Taghvaei et al. \cite{taghvaei2020optimal} apply an optimal transport framework to develop algorithms that feedback measurements in the evolution of the state probability density function, resulting in an elegant solution to the ensemble Kalman filtering problem. The resulting system is a controlled interaction between a finite number of particles that minimizes the Wasserstein distance between distributions over consecutive time steps. The work of Reich et al. and Taghvaei et al. overcame the problem ensemble Kalman filtering has with biases with finite particles and particle degeneracy. These early breakthroughs have led to considerable research in this field \cite{ singh2022inference, chen2021optimal, taghvaei2022optimal, corenflos2021differentiable, chen2021optimal1}.

\subsection*{Main Contributions}
In this paper, we introduce an alternative formulation for data assimilation, which seeks to determine an optimal map that combines both the prior distribution and the likelihood. The objective is to minimize the Wasserstein distance between the posterior error distribution and the Dirac distribution centered at the origin. Utilizing the Dirac distribution for error enables the representation of a perfect state estimate characterized by zero error and dispersion. We hypothesize that minimizing the Wasserstein distance between the posterior error and the Dirac distribution results in optimal information fusion and posterior state uncertainty. The Wasserstein distance, which defines a metric on the manifold of probability density functions (PDFs), is well-suited for formulating general estimation algorithms applicable to arbitrary uncertainty representations. It is important to note that the work presented in this paper does not rely on particle representations of uncertainty. Instead, it focuses on Gaussian and the mixture of Gaussian representations, which offers a more tractable and computationally efficient approach to the problem. 

Specifically, the following are the key contributions of the paper:\\[1mm]
i) In this paper, we introduce a novel general state-estimation formulation that aims to minimize the Wasserstein distance between the posterior error distribution and the Dirac distribution centered at the origin. The framework is distribution agnostic and can be applied to information fusion over general manifolds with arbitrary state uncertainty. \\[1mm]
ii) To validate the formulation, we demonstrate that the new framework successfully recovers the classical Kalman filter for linear systems characterized by Gaussian uncertainty.\\[1mm]
iii) The proposed Wasserstein filter recovers the classical result for the Gaussian Sum Filter. Notably, an analytical proof for the suboptimality condition of Gaussian Sum Estimators in minimum mean square error (MMSE) sense is presented using the new formulation.\\[1mm]
iv) Lastly, we showcase that by employing a nonlinear optimization technique, which is warm-started with the suboptimal solution, the accuracy of Gaussian Sum Estimators can be enhanced.

Key results are presented in Theorems \ref{kalman} and \ref{GMM}, along with supporting corollaries.

\section{MATHEMATICAL PRELIMINARIES}

\begin{definition} (\textbf{Wasserstein distance})
Given $z_{1}, z_{2} \in \mathbb{R}^{n}$, with $\mathcal{P}_{2}(p_1,p_2)$ containing the collection of all probability measures $p$ supported on the product space $\mathbb{R}^{2n}$, having a finite second moment, with first marginal $p_{1}$ and second marginal $p_{2}$. 2-Wasserstein distance, denoted as $W_2$, between two probability measures $p_{1},p_{2}$, is defined as
\label{Wassdefn}
\begin{align}
&W_2(p_{1},p_{2}) \triangleq \nonumber \\
&\left(\displaystyle\inf_{p\in\mathcal{P}_{2}(p_{1},p_{2})}\displaystyle\int_{\mathbb{R}^{2n}} \parallel z_{1}-z_{2}\parallel_{\ell_{2}\left(\mathbb{R}^{n}\right)}^{2} \: dp(z_{1},z_{2}) \right)
^{\frac{{1}}{2}}. \eqnlabel{W-dist}
\end{align}
\end{definition}
\begin{remark}
The Wasserstein distance can be interpreted as the cost for the Monge-Kantorovich optimal transportation plan\cite{villani2003topics}. This adds an intuitive meaning to the Wasserstein distance in terms of the least effort needed to morph from one distribution to another. The particular choice of $\ell_{2}$ norm with order 2 is motivated by \cite{halder2012further}. Further, one can prove (p. 208, \cite{villani2003topics}) that $W_2$ defines a metric on the manifold of PDFs, which makes it suitable for formulating distribution agnostic estimation algorithms.
\label{WassRemarkFirst}
\end{remark}
 
\begin{prop}
The Wasserstein distance $W_2$ between two multivariate Gaussians $\mathcal{N}(\mu_1, \Sigma_1)$ and $\mathcal{N}(\mu_2, \Sigma_2)$ in $\real^n$ is given by
\begin{align}   
&W_2^2\left(\mathcal{N}(\mu_1, \Sigma_1),\mathcal{N}(\mu_2, \Sigma_2)\right)= \nonumber \\
&\|\mu_1 - \mu_2\|^2+\tr{\Sigma_1+\Sigma_2-2\left(\sqrt{\Sigma_1}\Sigma_2\sqrt{\Sigma_1}\right)^{\frac{1}{2}}
}. \eqnlabel{wassdef}
\end{align}
\end{prop}

\begin{corollary} Square of 2-Wasserstein distance between Gaussian distribution $\mathcal{N}(\mu, \Sigma)$  and the Dirac delta function $\delta(x-\mu_c)$ is given by,
\begin{align}\eqnlabel{W_dirac}
W_2^2(\mathcal{N}(\mu, \Sigma),\delta(x-\mu_c))=\|\mu-\mu_c\|^2+\tr{\Sigma}.
\end{align}
\end{corollary}

\begin{proof}
Defining the Dirac delta function as (see e.g., p. 160-161, \cite{hassani1999mathematical})
\begin{align*}
   \delta(x-\mu_c)=\lim_{\mu\to \mu_c,\Sigma\to 0} \mathcal{N}(\mu,\Sigma),
\end{align*}
and substituting in \eqn{wassdef}, we get the result.
\end{proof}

The Wasserstein distance in \eqn{W_dirac} can also be written as,
\begin{align}
&W_2^2(\mathcal{N}(\mu, \Sigma),\delta(x-\mu_c)) =\|\mu-\mu_c\|^2+\tr{\Sigma},\nonumber \\
&={\tr{(\mu-\mu_c)(\mu-\mu_c)^T+\Sigma}} \eqnlabel{wass_distG}.
\end{align}



\section{WASSERSTEIN FILTER}
Here we consider the problem of updating the prior state estimate with available measurement data to arrive at the posterior state estimate. We assume $x\in\real^n$ is the true state. The prior estimate of $x$ is denoted by $x^-$, which is a random variable with associate probability density function $p_{x^-}(x^-)$. Similarly, the posterior estimate is denoted by $x^+$, also a random variable with associated probability density function  $p_{x^+}(x^+)$. We next assume that measurement $y\in\real^m$ is a function of the true state $x$, but corrupted by an additive noise $n$. It is modeled as 
\begin{align}
y := g(x) + n. \eqnlabel{meas}
\end{align} 
The noise is assumed to be a random variable with associated probability density function $p_n(n)$.

The errors associated with the prior and posterior estimates are defined as 
\begin{align}
e^- &:= x^- - x,\\
e^+ &:= x^+ - x,
\eqnlabel{errordef}
\end{align}
which are also random variables with probability density functions $p_{e^-}(e^-)$ and $p_{e^+}(e^+)$ respectively.

The objective of optimal estimation is to determine $x^+$ from $x^-$ and $y$ by minimizing some statistical quantity associated with $e^+$. More generally, we can define the optimal estimation problem as
\begin{align}
\min_{T} d(e^+) \eqnlabel{aopt}
\end{align}
where $T: (x^-,y) \mapsto x^+$, $d(e^+)$ is some cost function to be minimized. 

In optimal Wasserstein filtering, we define $d(e^+):= W^2_2(p_{e^+}(e^+),\delta(e^+))$, i.e. minimize the Wasserstein distance of the posterior-error's probability density function from the Dirac at the origin $\delta(e^+)$. The distribution $\delta(e^+)$ is the degenerate probability density function that represents the posterior error with zero mean and variance. The optimization therefore determines the map $T(\cdot,\cdot)$ that minimizes the Wasserstein distance between $p_{e^+}(e^+)$ and $\delta(e^+)$, resulting in the best posterior estimate of $x$ in this sense.

In the next subsections, we present a few specific cases of \eqn{aopt}, which are of engineering significance.
\subsection{Linear Measurement with Gaussian Uncertainty}
Given a linear sensor model $y:= Cx + n,n \sim \mathcal{N}(0,R)$ with Gaussian prior estimate $x^-$, we define a posterior linear map as $x^+ = T(x^-,y):= Gx^- + Hy$. It can be shown that $e^+$ is Gaussian distribution $\mathcal{N}(\mu_e^+, \Sigma_e^+)$, with
\begin{align}
\notag e^+ &=  x^+ - x, \\
\notag &= Gx^- + (HC-I)(x^- - e^-) + Hn,\\
&= \begin{bmatrix}(G+HC-I) & -(HC-I) & H \end{bmatrix}\begin{pmatrix}x^- \\ e^- \\ n \end{pmatrix}.
\end{align}
If $\Exp{n} = 0$ and $\Exp{e^-}=0$, then 
\begin{subequations}
\begin{align}
\mu_e^+ & := \Exp{e^+} = (G+HC-I)\mu^-, \eqnlabel{mue}\\
\Sigma_e^+ &:= \Var{e^+} - \mu_e^+{\mu_e^+}^T  \nonumber\\
&= \begin{bmatrix}(G+HC-I) & -(HC-I) & H \end{bmatrix}.\nonumber \\
&\Exp{\begin{pmatrix}x^- \\ e^- \\ n \end{pmatrix}\begin{pmatrix}x^- \\ e^- \\ n \end{pmatrix}^T}\begin{bmatrix}(G+HC-I)^T \\ -(HC-I)^T \\ H^T \end{bmatrix}& \nonumber \\
&- \mu_e^+{\mu_e^+}^T. 
\eqnlabel{sigmae}
\end{align}
\end{subequations}
In the case when, the joint probability density function associated with $e^+$ is Gaussian, denoted by $\mathcal{N}(\mu_e^+, \Sigma_e^+)$, the Wasserstein distance between $\mathcal{N}(\mu_e^+, \Sigma_e^+)$ and the Dirac delta function $\delta(e^+)$ is given by,
\begin{align}
W_2^2(\mathcal{N}(\mu_e^+, \Sigma_e^+),\delta(e^+))&=\tr{\mu_e^+{\mu_e^+}^T+\Sigma_e^+}.
\end{align}
Therefore, minimizing $W_2^2(\mathcal{N}(\mu_e^+, \Sigma_e^+),\delta(e^+))$ is the same as minimizing the posterior error variance, as it is done in Kalman filtering. The following theorem validates our argument that the proposed Wasserstein filter for a linear map between Gaussian distributed posterior and prior is the same as the Kalman filter. 
\begin{prop}(\textbf{Orthogonality Condition})
Consider a linear estimator with generic linear mapping between posterior and prior at the $k^\text{th}$ time step as 
\begin{align}
{x_k}^+ := G{x_k}^- + Hy_k.
\eqnlabel{linmap}
\end{align}
Let the posterior error be defined as ${e_k}^+ := {x_k}^+ -x_k$. The necessary conditions for a linear estimator to be optimal in the MMSE sense is given by
\begin{align}
    \Exp{{e_k^+}{x_k^-}^T} &= 0 \eqnlabel{optc1} \\
    \Exp{{e_k^+}{y_k}^T} &= 0. \eqnlabel{optc2}
\end{align}
This implies that at any time step $k$, the posterior error must be orthogonal to the prior state and current measurement.
\label{Orthog}
\end{prop}

\begin{proof}
Given the linear estimator \eqn{linmap}, the cost function for MMSE optimization can be written as
\begin{align}
    \notag J &= \textbf{tr}\left\{\Exp{{e_k^+}{e_k^+}^T}\right\} \\
    \notag   &= \textbf{tr}\left\{\Exp{{e_k^+}{x_k^-}^T}G^T + \right. \\
    & \left.{} \Exp{{e_k^+}{y_k}^T}H^T - \Exp{{e_k^+}{x_k}^T}\right\}.\eqnlabel{Jlin}
\end{align}
Applying necessary first-order optimal conditions,$\frac{\partial J}{\partial G}=0; \frac{\partial J}{\partial H}=0$, \eqn{optc1} and \eqn{optc2} are obtained.
\end{proof}

\begin{corollary} 
Given an optimal linear estimator in the MMSE sense, consider a linear time-invariant state propagation model $x_{k+1} = A{x}_k + w_k$, with propagation noise $w_k \sim \mathcal{N}(0,Q)$. It can be shown that prior state and error estimate at any time step $k$, $x_k^-$ and $e_k^-:=x_k^--x$ respectively, are orthogonal.
\label{cor1}
\end{corollary}

\begin{proof}
Following the given definitions, we simplify $\Exp{{e_k^-}{x_k^-}^T}$ as:
\begin{align}
    \notag&\Exp{{e_k^-}{x_k^-}^T} \\
    \notag&\quad = \Exp{{({x_k^-}-x_k)}{x_k^-}^T} \\
    \notag&\quad = \Exp{{({A{x}_{k-1}^+}-{A{x_{k-1}}} -w_{k-1})}{A{x}_{k-1}^+}^T} \\
    \notag&\quad = \Exp{{({A{e}_{k-1}^+} -w_{k-1})}{(A{x}_{k-1}^+)}^T} \\
    \notag&\quad = \Exp{{({A{e}_{k-1}^+})}{(A{x}_{k-1}^+)}^T} - \Exp{({w_{k-1})}{(A{x}_{k-1}^+)}^T} \\
    &\quad = A{\Exp{{{e}_{k-1}^+}{x_{k-1}^+}^T}} A^T - {\Exp{{w_{k-1}}{{x}_{k-1}^+}^T}}A^T \eqnlabel{Exe}
\end{align}
Using the uncorrelated property of white noise, \eqn{Exe} reduced to
\begin{align}
    \notag&\Exp{{{e}_k^-}{{x}_k^-}^T} \\
    \notag&\quad = A{\Exp{{{e}_{k-1}^+}{x_{k-1}^+}^T}} A^T - {\Exp{{w_{k-1}}}\Exp{{{x}_{k-1}^+}^T}}A^T \\
    &\quad = A{\Exp{{{e}_{k-1}^+}{x_{k-1}^+}^T}} A^T \eqnlabel{Exe1}.
\end{align}
Following the results of proposition \ref{Orthog}, the remaining term in \eqn{Exe1} drops off. Thus,
\begin{align}
    \Exp{{{e}_k^-}{{x}_k^-}^T} = 0 \eqnlabel{Exe2}.
\end{align}
Eq.\eqn{Exe2} establishes the orthogonal relation between prior state estimate and prior error at any time step $k$.
\end{proof}

\begin{theorem}
Given linear sensor model $y := Cx + n$, with measurement noise $n \sim \mathcal{N}(0,R)$, and defining a linear map between posterior and prior state estimates as $x^+ = T(x^-,y):= Gx^- + Hy$, minimization of Wasserstein distance $d(e^+):= W^2_2(p_{e^+}(e^+),\delta(e^+))$ results in the Kalman filter with $G^\ast := I-H^\ast C$ and $H^\ast := \Sigma_e^-C^T(C\Sigma_e^-C^T+R)^{-1}$.
\label{kalman}
\end{theorem}

\begin{proof} 
Using \eqn{wass_distG}, the optimization problem can be formally stated as:
\begin{align}
\min_{T} \tr{\mu_e^+{\mu_e^+}^T+\Sigma_e^+}. \eqnlabel{Kopt}
\end{align}
Substituting \eqn{mue} and \eqn{sigmae} in \eqn{Kopt}, we get the cost function $J$ as:
\begin{align}
\footnotesize
J &= \textbf{tr}\Bigg\{\begin{bmatrix}(G+HC-I) & -(HC-I) & H \end{bmatrix}\times \nonumber\\
&\Exp{\begin{pmatrix}x^- \\ e^- \\ n \end{pmatrix}\begin{pmatrix}x^- \\ e^- \\ n \end{pmatrix}^T}\begin{bmatrix}(G+HC-I)^T \\ -(HC-I)^T \\ H^T \end{bmatrix}\Bigg\}
\eqnlabel{cost}
\end{align}
With $\Sigma_e^-=\Var{e^-}$, we expand the cost function and further simplify the expression using properties of measurement white noise as:
\begin{align}
    \notag J &= \textbf{tr}\{(G+HC-I){\Var{x^-}}{(G+HC-I)^T} - \\
    \notag &(G+HC-I){\Exp{x^{-}{e^{-}}^T}}{(HC-I)^T} - \\
    \notag &(HC-I){\Exp{x^{-}{e^{-}}^T}}{(G+HC-I)^T} + \\
    &(HC-I){\Sigma_e^-}{(HC-I)^T} + HR{H}^T\} 
\end{align}
Enforcing the first order optimal conditions $\frac{\partial J}{\partial G}=0; \frac{\partial J}{\partial H}=0$, we get:
\begin{subequations}
\begin{align}
\notag &(G+HC-I){\Var{x^-}}\\
 &- (HC-I){\Exp{x^{-}{e^{-}}^T}} = 0
\eqnlabel{FOC1}\\
\notag &(HC-I){\Sigma_e^-}{C^T} -(G+HC-I){\Exp{x^{-}{e^{-}}^T}}{C^T} \\
& + HR = 0
\eqnlabel{FOC2}
\end{align}

\end{subequations}
\end{proof}
An analytical solution for \eqn{FOC1}, \eqn{FOC2} can be obtained by setting $G+HC-I=0$. Following the assumption, \eqn{FOC1} reduced to ${\Exp{x^{-}{e^{-}}^T}} = 0$. Indeed, this result holds for linear estimator optimal in MMSE sense as shown in corollary \ref{cor1}.  Solving \eqn{FOC2}, we obtain corresponding stationary point,
\begin{subequations}
\begin{align}
H^\ast &= {\Sigma_{e}^-}{C^T}(C{\Sigma_{e}^-}C^T + R)^{-1}
\eqnlabel{H}\\
G^\ast &= I - {\Sigma_{e}^-}{C^T}(C{\Sigma_{e}^-}C^T + R)^{-1}C
\eqnlabel{G}
\end{align}
\end{subequations}
The positive definiteness of the Hessian matrix $\textbf{H} = (C{\Sigma_{e}^-}C^T + R)$ affirms the optimality of the obtained stationary point. Further, posterior error mean $(\mu_e^+)$, covariance $(\Sigma_e^+)$, and measurement update are obtained as:
\begin{subequations}
\begin{align}
\mu_e^+ &= 0, \eqnlabel{mue+}\\
\notag \Sigma_e^+ &= (HC-I){\Sigma_{e}^-}(HC-I)^T + HRH^T,  \\
\notag &= H(C{\Sigma_{e}^-}C^T + R)H^T + {\Sigma_{e}^-} -  \\
\notag &\qquad HC{\Sigma_{e}^-} - {\Sigma_{e}^-}{C^T}{H^T}, \\
&= {\Sigma_{e}^-} - {\Sigma_{e}^-}{C^T}(C{\Sigma_{e}^-}C^T + R)^{-1}C{\Sigma_{e}^-} 
\eqnlabel{sigmae+} \\
x^{+} &= x^{-} + {\Sigma_{e}^-}{C^T}(C{\Sigma_{e}^-}C^T + R)^{-1}(y - Cx^{-}).
\eqnlabel{update}
\end{align}
\end{subequations}
Therefore, for linear measurements and Gaussian uncertainty models, the Wasserstein filter is equivalent to the Kalman filter.

\subsection{Mixture of Gaussian (MoG) Prior}

In the previous section, we considered Gaussian prior estimate with a linear measurement model with Gaussian noises\eqn{meas}. In this section, we consider the case of fusing multimodal Gaussian mixture prior i.e.,  

\begin{align}
    x^-\sim  \sum_i^M\lambda_i^- \mathcal{N}(\mu_i^-,\Sigma_i^-), \eqnlabel{MoG1}
\end{align}

where $0\leq\lambda_i\leq1$ and $$\sum_i^M\lambda_i=1,$$
with a linear measurement model containing Gaussian measurement noise.
 Alternatively, the $M$-order Gaussian mixture prior can be expressed using the generative model approach with a random vector $\boldsymbol{\beta}^-:= [\beta_1^-,..,\beta_M^-], $ indicating on each generating Gaussian, with pdf given by \cite{Bilik2005},
\begin{align}
\rho(\boldsymbol{\beta}^-)&= \sum_{j=1}^M {\lambda_j^-}\delta({\beta_j^-}-1) \\
\notag \text{where,} \\
\beta_j^- &= 
    \begin{cases}
        1, & \text{if } x^-\text{generated by } j^{th} \text{ Gaussian}\\
        0, & \text{otherwise}
    \end{cases} \\
\notag \text{and,} \\
P&(\beta_j^- = 1) := P(\beta_j^-) = \lambda_j^-.
\eqnlabel{indicatorpdf}
\end{align}
We introduce a latent random variable ${z_j^-} := (x^-|\beta_j^-) \sim \mathcal{N}(\mu_j^-, \Sigma_j^-)$ characterizing the $j^{th}$ Gaussian node. Using the generative model approach, the multimodal Gaussian prior can thus be written as
\begin{align}
    x^-\sim\sum_i^M P(\beta_i^-)\rho(z_i^-).
\end{align}
With linear Gaussian measurement, we can assume the posterior distribution to be of the form
\begin{align}
    \notag x^+&\sim  \sum_i^M\lambda_i^+ N(\mu_i^+,\Sigma_i^+)\\
    \notag \text{or,}\\
            x^+&\sim \sum_i^M P(\beta_i^+)\rho(z_i^+)\\
    \notag \text{where,} \\
    \notag z_i^+ &:= (x^+|\beta_i^+)\sim \mathcal{N}(\mu_i^+, \Sigma_i^+).
\end{align}
This interpretation of a mixture of Gaussian(MoG) suggests that there are $M$  latent variables representing each Gaussian. Fusion of a MoG prior with Gaussian likelihood results in posterior MoG with the same number of nodes. Each prior Gaussian node is one-to-one mapped with another posterior node and we assume that this mapping is linear, but the map's parameters are dependent upon the mode and the measurement values. The linear mapping between each corresponding $j^{th}$ Gaussian node of posterior and prior is given by 
\begin{align}
    z_i^+ = G_i z_i^- + H_i y
    \eqnlabel{map2}
\end{align}

\noindent As in \eqn{errordef}, we define the posterior error variable as,
\begin{align}
    e^+ = x^+ - x,
\end{align}
where $x$ is the deterministic true state variable. $e^{+}$ is a shifted version of the distribution of $x^+$ with pdf given by
\begin{align}
    e^+ &\sim \sum_i^M \lambda_i^+ \mathcal{N}(\mu_i^+ - x, \Sigma_i^+) \eqnlabel{posteriorerr}\\
    \notag \text{ or},\\
    e^+ &\sim \sum_i^M P(\beta_i^+)\rho(e_i^+)\\
    \notag \text{where},\\
    e_i^+ &:= z_i^+-x. \eqnlabel{ei}
\end{align}
\begin{prop}(\textbf{$W_2$ between MoG and Dirac reference PDF})
The Wasserstein distance between a multimodal Gaussian Mixture Distribution and the Dirac delta at origin is given by,
\begin{align}
    W_2^2 = \sum_i^M\lambda_i W_i^2.
\end{align}
\label{MoGP}
\end{prop}
\begin{proof}
\textit{refer Thm. 2 \cite{LEE2015341}}
\end{proof}
\begin{remark}
    Given the Gaussian mixture prior alongside linear Gaussian measurement, the Wasserstein distance between posterior error $e^+$ and Dirac delta distribution centered at the origin is given by,

    \begin{align}
        \notag W_2^2 &= \sum_i^M \lambda_i^+ W_i^2(N(\mu_i^+-x, \Sigma_i^+),\delta(y))\\
        &\notag  = \sum_i^M tr(\lambda_i^+ \{(\mu_i^+ - x)({\mu_i^+ - x})^T+\Sigma_i^+\}) \\
        & = \sum_i^M tr(\lambda_i^+ \Exp{{e_i^ +}{e_i^+}^T}) \eqnlabel{MoGW}
    \end{align}
\end{remark}
Using \eqn{map2}, \eqn{ei} in \eqn{MoGW} we get,
\begin{align}
    \notag W_i^2 &= \textbf{tr}\{\begin{bmatrix}(G_i+H_i C-I) & -(H_i C-I) & H_i \end{bmatrix}. \nonumber\\
&\Exp{\begin{pmatrix}z_i^- \\ e_i^- \\ n \end{pmatrix}\begin{pmatrix}z_i^- \\ e_i^- \\ n \end{pmatrix}^T}\begin{bmatrix}(G_i+H_i C-I)^T \\ -(H_i C-I)^T \\ H_i^T \end{bmatrix}\} 
\eqnlabel{eim}
\end{align}
Wasserstein distance minimization for MoG prior results in the following optimization problem:
\begin{align}
    \min_{\{\lambda_i, G_i, H_i\}_{\forall i \in [1,M]}} \sum_i^M \textbf{tr}\{\lambda_i^+ \Exp{{e_i^ +}{e_i^+}^T}\} 
    \eqnlabel{Kopt2}
\end{align}
subject to
$$\sum_i^M \lambda_i^+ = 1,\lambda_i \geq 0.$$
Although the objective function is non-convex, there exists a convex upper bound 
\begin{align}
\sum_i^M \textbf{tr}\{\lambda_i^+ \Exp{{e_i^ +}{e_i^+}^T}\} \leq \sum_i^M \textbf{tr}\{\Exp{{e_i^ +}{e_i^+}^T}\}, \eqnlabel{ubound}
\end{align}
which can be conveniently optimized. It can be shown that bound minimization leads to Gaussian sum filter (GSF)\cite{Bilik2005}.
\begin{theorem}
Given the prior distribution represented by MoG alongside linear sensor model $y:= Cx + n; n \sim \mathcal{N}(0, R)$, we assume linear maps between corresponding posterior and prior latent variables $z_{i}s$ as $z_i^+ = T_i(z_i^-,y):= G_i z_i^- + H_i y$. The minimization of the upper bound of Wasserstein distance $d(e^+) \leq \sum_i^M W^2_2(p_{e_i^+}(e_i^+),\delta(e_i^+))$ given in \eqn{ubound} results in the Gaussian sum filter (GSF) with $G_i^\ast:= I-H_i^\ast C$ and $H_i^\ast:= \Sigma_{i}^-C^T(C\Sigma_{i}^-C^T+R)^{-1}$.
\label{GMM}
\end{theorem}
\begin{proof}
    Relaxed minimization problem can be derived from \eqn{Kopt2} as,
    \begin{align}
        \sum_i^M \textbf{tr}\{\lambda_i^+ \Exp{{e_i^ +}{e_i^+}^T}\} \leq \sum_i^M \textbf{tr}\{\Exp{{e_i^ +}{e_i^+}^T}\}.
        \eqnlabel{rlxmin}
    \end{align}
    Since the cost function is a sum of squares with decoupled associated parameters, each term can be independently optimized. Following the expansion of the $i^{th}$ term using \eqn{eim}, we have
    \begin{align}
    \notag J_i &= \textbf{tr}\{(G_i+H_iC-I){\Var{x_i^-}}{(G_i+H_iC-I)^T} \\
    \notag &-(G_i+H_iC-I){\Exp{x_i^{-}{e_i^{-}}^T}}{(H_iC-I)^T} - \\
    \notag &(H_iC-I){\Exp{x_i^{-}{e_i^{-}}^T}}{(G_i+H_iC-I)^T} + \\
    &(H_iC-I){\Sigma_{i}^-}{(H_iC-I)^T} + H_iR{H_i}^T\}.
    \eqnlabel{costi}
    \end{align}
    Considering the linear map between Gaussian marginalized prior and posterior with minimum mean squared marginalized error as the objective function, the orthogonality conditions mentioned in proposition \ref{Orthog} holds well, i.e.
    \begin{align}
        \Exp{{{e_i}_k^-}{{x_i}_k^-}^T} = 0 \eqnlabel{Exe3}.
    \end{align}
    The objective function minimization follows the same approach detailed in theorem \ref{kalman}, thus leading to the following stationary point:
    \begin{subequations}
        \begin{align}
            H_i^\ast &= {\Sigma_{i}^-}{C^T}(C{\Sigma_{i}^-}C^T + R)^{-1}
            \eqnlabel{Hi}\\
            G_i^\ast &= I - {\Sigma_{i}^-}{C^T}(C{\Sigma_{i}^-}C^T + R)^{-1}C
            \eqnlabel{Gi}
        \end{align}
    \end{subequations}
    A reasonable estimate of mixing coefficients $\lambda_is$ can be obtained using the Bayesian approach, \cite{AndersonBD}
    \begin{align}
        \lambda_i^+ = \frac{{\lambda_i^-}\mathcal{N}(y;C\mu_i^-, C^T\Sigma_i^-C+R)}{\sum_j^M\lambda_j^-\mathcal{N}(y;C\mu_j^-, C^T\Sigma_j^-C+R)}
        \eqnlabel{lamu}
    \end{align}    
\end{proof}
\begin{algorithm}
\caption{GSF Measurement Update}\label{GSF}
\begin{algorithmic}
\State $\textbf{Given:}\ \mu_i^-,\Sigma_i^-\ \forall i=1,..M$
\State
\State $\textbf{Posterior Parameters:}\ \{H_{i},\lambda_{i}^+\}$
\State $\  H_{i} = {\Sigma_{i}^-}{C^T}(C{\Sigma_{i}^-}C^T + R)^{-1} $
\State $\  \lambda_{i}^+ = \frac{{\lambda_i^-}\mathcal{N}(y;C\mu_i^-, C^T\Sigma_i^-C+R)}{\sum_j^M\lambda_j^-\mathcal{N}(y;C\mu_j^-, C^T\Sigma_j^-C+R)}$
\State
\State $\textbf{Posterior Update:}$
\State $\ \mu_i^{+} = \mu_i^{-} + H_i(y - C\mu_i^{-})$
\State $\ {\Sigma_{i}^+} = {\Sigma_{i}^-} - H_iC{\Sigma_{i}^-}$
\State
\State $\textbf{Posterior GMM:}\ x^+\sim\sum_i^M\lambda_i^+ N(\mu_i^+,\Sigma_i^+)$
\end{algorithmic}
\end{algorithm}
It should be noted that obtained $H_{i}^\ast,\ \lambda_i^+ $ for GSF are optimal only for the relaxed minimization problem in \eqn{rlxmin} and sub-optimal for exact optimization problem in \eqn{Kopt2}. Considering the sub-optimal nature of the solution, an intuitive idea is to perform a local search in the vicinity of $(H_{i}^\ast,\ \lambda_i^+),$ with the goal of improvement over the GSF. The cost function in \eqn{costi} can be simplified using \eqn{Exe3} and trivial first order optimal condition $G_i+H_iC-I=0$,
    \begin{align}
        \notag J_i &=  \\
        &\qquad \textbf{tr}\{(H_iC-I){\Sigma_{i}^-}{(H_iC-I)^T} + H_iR{H_i}^T\}.
        \eqnlabel{costi2}
    \end{align}
Using \eqn{costi2}, we formulate the nonlinear Gaussian sum filter (nGSF) which solves the following nonlinear constrained optimization with sub-optimal solutions in \eqn{Hi}, \eqn{Gi}, and \eqn{lamu} as initial guess,
\begin{align}
    \notag \hat{J} &= min_{\{\lambda_i^+,H_i\}}\\
    &\sum_i^M\lambda_i^+\textbf{tr}\{(H_iC-I){\Sigma_{i}^-}{(H_iC-I)^T} + H_iR{H_i}^T\},
    \eqnlabel{costmod}
\end{align}
subject to
$$\sum_i^M\lambda_i^+ = 1; \lambda_i^+ \geq 0.$$

\begin{algorithm}
\caption{nGSF Measurement Update}\label{nGSF}
\begin{algorithmic}
\State $\textbf{Given:}\ \mu_i^-,\Sigma_i^-\ \forall i=1,..M$
\State
\State $\textbf{Initial Guess:}\ \{H_{i0},\lambda_{i0}^+\}$
\State $\  H_{i0} = {\Sigma_{i}^-}{C^T}(C{\Sigma_{i}^-}C^T + R)^{-1} $
\State $\  \lambda_{i0}^+ = \frac{{\lambda_i^-}\mathcal{N}(y;C\mu_i^-, C^T\Sigma_i^-C+R)}{\sum_j^M\lambda_j^-\mathcal{N}(y;C\mu_j^-, C^T\Sigma_j^-C+R)}$
\State
\State $\textbf{Nonlinear Optimization:}$
\State $\ \textit{Cost:}\ \min_{\{\lambda_i^+, H_i\}_{\forall i \in [1,M]}} \sum_i^M \textbf{tr}\{\lambda_i^+ \Exp{{e_i^ +}{e_i^+}^T}\}$
\State $\ \textit{Constraints:}\ \sum_i^M\lambda_i^+ = 1; \lambda_i^+ \geq 0$
\State
\State $\textbf{Posterior Update:}$
\State $\ \mu_i^{+} = \mu_i^{-} + H_i^{\ast}(y - C\mu_i^{-})$
\State $\ {\Sigma_{i}^+} = {\Sigma_{i}^-} - H_i^{\ast}C{\Sigma_{i}^-}$
\State
\State $\textbf{Posterior GMM:}\ x^+\sim\sum_i^M\lambda_i^+ N(\mu_i^+,\Sigma_i^+)$
\end{algorithmic}
\end{algorithm}

\section{EXAMPLE}
We consider the classical duffing oscillator as an example to compare the performance of the nonlinear Gaussian sum filter (nGSF) with baseline GSF. The nonlinear dynamics in terms of two states $x = [x_1,x_2]'$ is given by:
\begin{align}
    \dot{x}_1=x_2,\ \dot{x}_2=-x_1-0.25x_2-x_1^3.
\end{align}

We propagate Gaussian distributed uncertainty in the initial condition, i.e. $\Delta_{x_0}\sim\mathcal{N}(0, I)$, using particle ensemble. Next, we obtain the Gaussian mixture prior by fitting the MoG over the propagated point cloud. It is evident from Fig.\ref{fig:result0} that despite the Gaussian uncertainty in the initial condition, the propagated distribution is non-Gaussian. In the simulated cases, the prior state distribution at each time step is approximated by a Gaussian mixture with 10 components.  

Considering a linear scalar measurement  corrupted by Gaussian noise, i.e.
\begin{align}
    y = x + v_k,\ v_k \sim \mathcal{N}(0,0.1),
\end{align}
we obtain the posterior MoG using the GSF and nGSF algorithms. Random sampling from the obtained posterior gives us a prior point cloud for the next time step. 
The estimation performance for baseline GSF and nGSF is presented in Fig.\ref{fig:result1a} and Fig. \ref{fig:result1b}.
\begin{figure}[h]
\includegraphics[width=7.5cm]{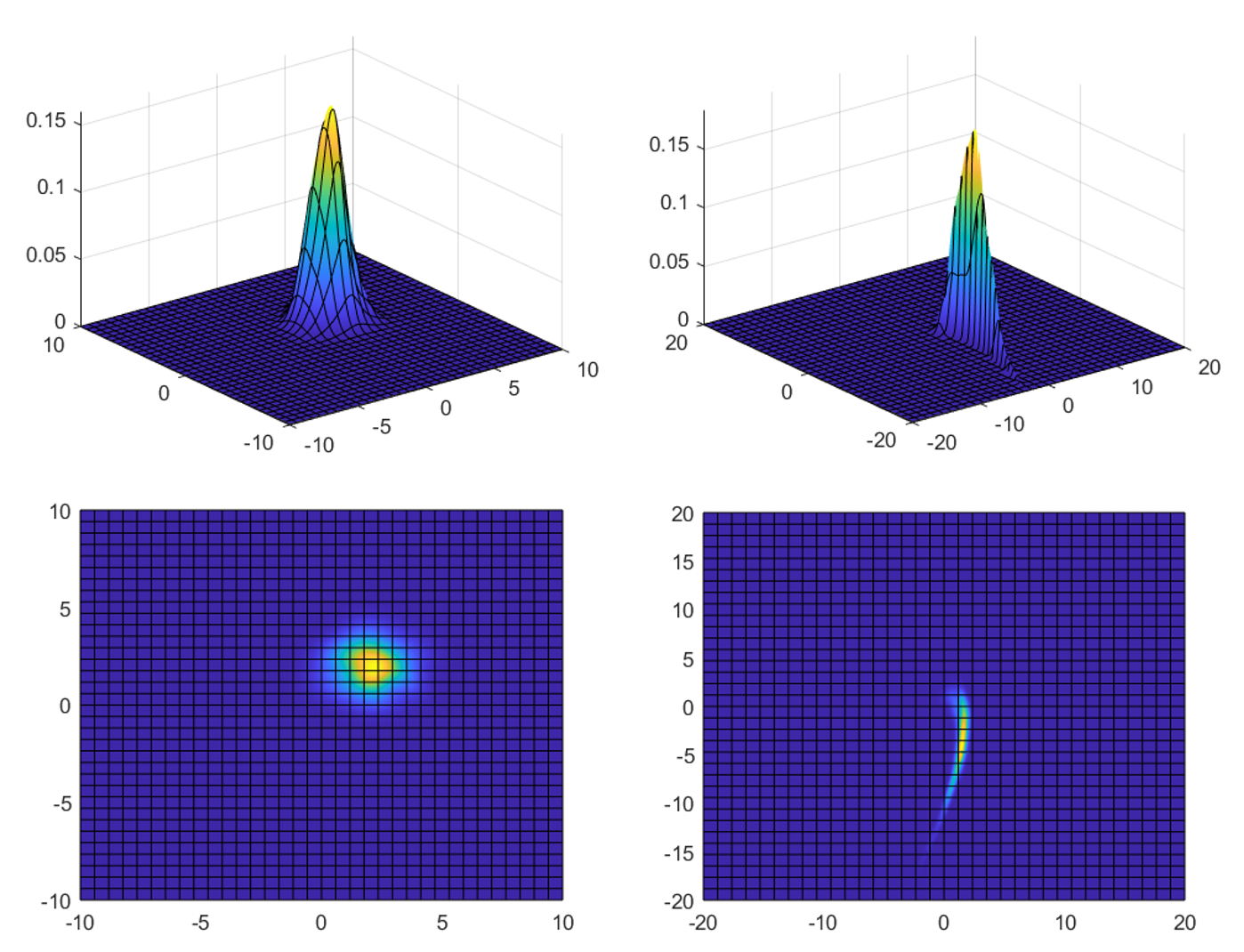}
\caption{Initial and propagated state uncertainty after 1 time step (0.5 seconds).}
\label{fig:result0}
\end{figure}
\begin{figure}[h]
\includegraphics[width=7.5cm]{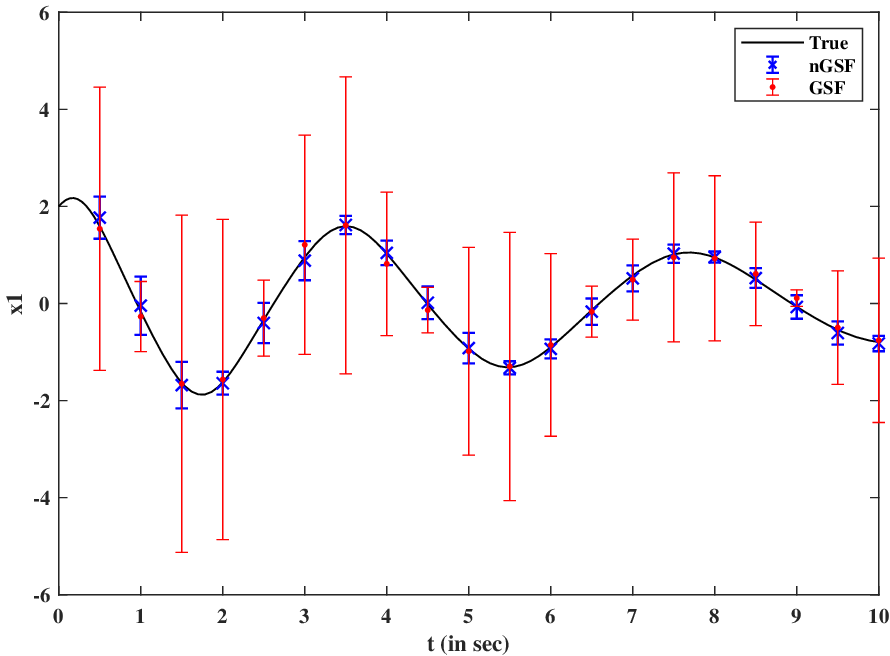}
\caption{Position ($x_1$) estimates with filter sample time: 0.5s.}
\label{fig:result1a}
\end{figure}
\begin{figure}[h]
\includegraphics[width=7.5cm]{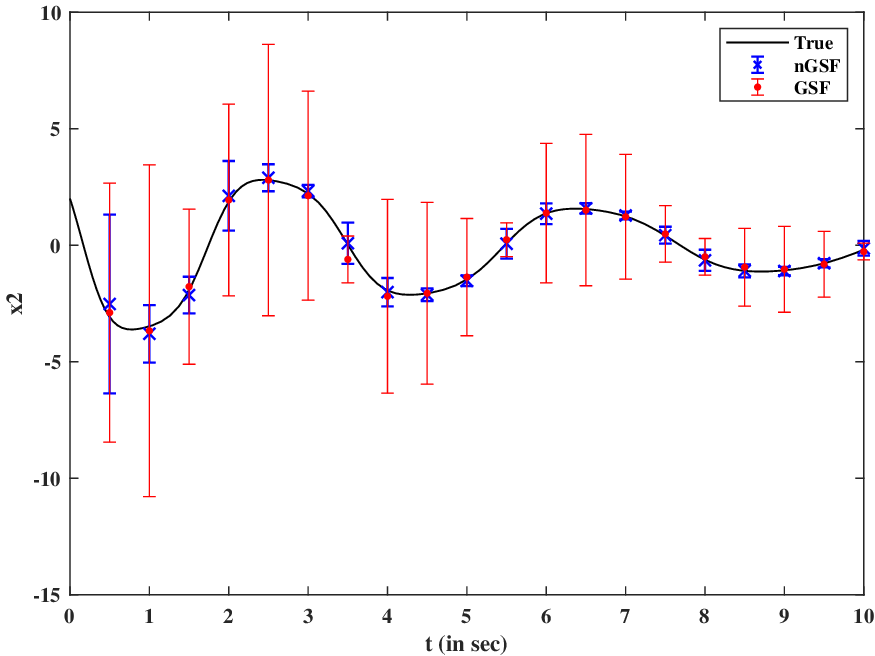}
\caption{Velocity ($x_2$) estimates with filter sample time: 0.5s.}
\label{fig:result1b}
\end{figure}
We observe the performance enhancement in terms of reduced error variance in the case of nGSF compared to GSF. 
\section{CONCLUSIONS}
A novel distribution agnostic estimation framework based on Wasserstein distance is formulated in this work. It is demonstrated that for linear Gaussian uncertainty, the optimal filter in the Wasserstein metric is essentially the Kalman filter. Moreover, the Gaussian Sum Filter is derived in the Wasserstein framework with explicit proof of sub-optimality in the MMSE sense. Notably, other approaches do not provide such explicit guarantees. Furthermore, the proposed nGSF is shown to enhance the accuracy of the Gaussian sum estimator using a nonlinear optimization technique, at the cost of additional computation. In our future work, we will relax the linear Gaussian measurement restriction and incorporate nonlinear measurement models with non Gaussian noise.

\section{ACKNOWLEDGMENTS}
This work is supported by AFOSR grant FA9550-22-1-0539 with Dr. Erik Blasch as the program director.


\bibliographystyle{unsrt}
\bibliography{references}





\end{document}